\documentclass[11pt]{article}

\usepackage{amsfonts,amssymb,amsthm,amsmath,comment}
\usepackage{hyperref}

\usepackage{bm}
\usepackage{mathtools}
\usepackage{complexity}
\usepackage{thm-restate}

\usepackage[linesnumbered,ruled,vlined]{algorithm2e}

\usepackage[margin=1in]{geometry}
\usepackage{bm}
\usepackage{fancybox}
\usepackage{color}
\usepackage{latexsym}
\newtheorem{theorem}{Theorem}
\newtheorem{lemma}[theorem]{Lemma}

\newtheorem{proposition}[theorem]{Proposition}
\theoremstyle{definition}

\newtheorem{remark}[theorem]{Remark}

\DeclareMathOperator{\DS}{DS}

\DeclareMathOperator{\ADS}{ADS}  
  
\newcommand{\F}{\mathbb{F}}

\newcommand{\ang}[1]{{\left\langle{#1}\right\rangle}}

\DeclareMathOperator{\Cay}{Cay}

\let\Mod\relax
\DeclareMathOperator{\Mod}{mod}

\renewcommand{\le}{\leqslant}
\renewcommand{\leq}{\leqslant}

\renewcommand{\angle}[1]{\langle #1 \rangle}

\begin{document}
\title{The Parallel Dynamic Complexity of the Abelian Cayley Group
  Membership Problem}
%
%\titlerunning{Dynamic Complexity of Abelian CGM}
\author{V. Arvind\thanks{The Institute of Mathematical Sciences
    (HBNI), Chennai, India, and Chennai Mathematical Institute,
    Chennai, India \texttt{email: arvind@imsc.res.in}} \and Samir
  Datta\thanks{Chennai Mathematical Institute, Chennai, India
    \texttt{email: sdatta@cmi.ac.in}} \and Asif Khan\thanks{Chennai
    Mathematical Institute, Chennai, India\texttt{ email:
      asifkhan@cmi.ac.in}} \and Shivdutt Sharma\thanks{Indian
    Institute of Information Technology, Una, India \texttt{email:
      shiv.sharma@alumni.iitgn.ac.in}} \and Yadu Vasudev\thanks{Indian
    Institute of Technology Madras, Chennai, India \texttt{email:
      yadu@cse.iitm.ac.in}} \and Shankar Ram Vasudevan\thanks{Chennai
    Mathematical Institute, Chennai, India \texttt{email:
      shankarram@cmi.ac.in}}}

\maketitle

% typeset the header of the contribution
%
%\begin{comment}
\begin{abstract}
  Let $G$ be a finite group given as input by its multiplication
  table. For a subset $S\subseteq G$ and an element $g\in G$ the
  \emph{Cayley Group Membership Problem} ($\CGM$) is to check if $g$
  belongs to the subgroup generated by $S$. While this problem is
  easily seen to be in polynomial time, pinpointing its parallel
  complexity has been of research interest over the years. Barrington
  et al \cite{BKLM} have shown that for abelian groups the $\CGM$
  problem can be solved in $O(\log\log |G|)$ parallel time.  In this
  paper we further explore the parallel complexity of the abelian
  $\CGM$ problem, with focus on the dynamic setting: the generating
  set $S$ changes with insertions and deletions and the goal is to
  maintain a data structure that supports efficient membership queries
  to the subgroup $\angle{S}$. We obtain the following results:
 \begin{enumerate}
 \item First, we consider the more general problem of Monoid
   Membership, where $G$ is a monoid input by its multiplication
   table. When $G$ is a \emph{commutative monoid} we show there is a
   deterministic dynamic $\AC^0$ algorithm\footnote{Equivalently, a
    constant time parallel algorithm using polynomially many
     processors.} for membership testing that supports $O(1)$
   insertions and deletions in each step.
 \item Building on the previous result we show that there 
   is a dynamic randomized $\AC^0$ 
   algorithm for abelian $\CGM$ that supports 
   $\polylog(|G|)$
   insertions/deletions to $S$ in each step.
 \item If the number of insertions/deletions is at most $O(\log n/\log\log n)$
   then we obtain a deterministic dynamic $\AC^0$ algorithm for abelian $\CGM$.
 \item Applying these algorithms we obtain analogous results for the
   dynamic abelian Group Isomorphism.
 %\item 
% \item In the static setting, we give an $O(\log\log n)$ parallel
%   algorithm for computing an \emph{independent generating set} for
%   $G$.
 \end{enumerate}
 We also consider the problem when the multiplication table for
 $G$ is dynamic.
\end{abstract}
%\end{comment}

%\begin{center}
%\Large\textbf{APPENDIX}
%\end{center}

\section{Introduction}

The main algorithmic problem of interest in this paper, is the
\emph{Cayley Group Membership Problem} ($\CGM$):~ Given as input a
finite group $G$ by its multiplication table (also known as its Cayley
table), a subset $S\subseteq G$ and an element $g\in G$, test if
$g\in\angle{S}$, where $\angle{S}$ is the subgroup of $G$ generated by
the elements in $S$.

The $\CGM$ problem was brought into focus by the work of Barrington et
al \cite{BKLM} which raises intriguing questions about its parallel
complexity.

\paragraph*{\textbf{Background}}
Membership testing in finite groups is well-studied
\cite{Seressbook}. Its computational complexity significantly depends
on how $G$ is given as input and on its elements' representation. For
example, if the elements of $G$ are represented as permutations on
$[n]=\{1,2,\ldots,n\}$, then $G$ is a subgroup of $S_n$, the group of
all permutations on $[n]$. A natural compact description of $G$ as
input is by a generating set, as every finite group $G$ has a
generating set of size at most $\log |G|$. In this form, membership
testing in \emph{permutation groups} has been studied since the
1970's, pioneered by the work of Sims \cite{Sims,Seressbook}. There
are efficient polynomial (in $n$) time algorithms for the problem as
well as parallel algorithms for it. The problem is in $\NC$: it can be
be solved in $\polylog(n)$ time with polynomially many processors
\cite{BLS}. On the other hand, $G=\angle{a}$ could be a cyclic
subgroup, generated by $a$, of $\F_p^*$, the multiplicative group of
the finite field $\F_p$, where the prime $p$ is given in
binary. Testing if $b\in\angle{a}$, for $b\in\F_p^*$ is considered
computationally hard. The search version of solving for $x$ such that
$a^x=b$ is the \emph{discrete log} problem, widely believed
intractable for random primes $p$.

\paragraph*{\textbf{Cayley Table Representation}}
The Cayley table representation of $G$, in contrast, makes the $\CGM$
problem algorithmically easy: we can define a graph $X=(V,E)$ with
$V=G$ as vertex set and $(x,y)\in E$ if $xs=y$ or $ys=x$ for a
generator $s\in S$. Then, $g$ is in the subgroup generated by $S$ if
and only if the vertex $g$ is reachable from the identity element of
$G$. Indeed, this is an instance of undirected graph reachability
which has a polynomial time and even a \emph{deterministic logspace}
algorithm due to Reingold \cite{Rei04}.  That is, $\CGM$ is in the
complexity class $\L$ (which is contained in $\P$). Since it is in
$\L$ it is also in the circuit complexity class $\AC(\log n)=\AC^1$,
which means it has $\log$-depth polynomial-size circuits of unbounded
fanin. Equivalently, this means $\CGM$ has a logarithmic time CRCW
PRAM algorithm (we will define the relevant parallel complexity
classes in Section~\ref{sec2}).

\paragraph*{\textbf{Parallel Complexity of $\CGM$ and
Group Isomorphism}}

Chattopadhyay, Tor\'an and Wagner \cite{CT12} have shown that the
Group Isomorphism problem of checking if two groups $G_1$ and $G_2$
given as input by their multiplication tables are isomorphic can be
solved by quasipolynomial size constant-depth circuits. While the
question whether or not Group Isomorphism is in $\P$ is open and is
intensely studied in recent times \cite{Sun,GQ21a,GQ21b}, the above
parallel complexity upper bound implies that even Parity is not
reducible to Group Isomorphism!  Similarly, Fleischer has observed,
based on \cite{CT12} that the $\CGM$ problem can also be solved by
quasipolynomial size constant depth circuits. Since there is no
hardness result for $\CGM$, pinpointing its parallel complexity is an
interesting question.

As already mentioned, Barrington et al \cite{BKLM} have made nice
progress showing that $\CGM$ for abelian groups is in $\AC(\log\log
n)$. Indeed, since the resulting circuits are dlogtime uniform, the
upper bound is $\FOLL$ (which means first-order definable with
$\log\log n$ quantifier depth, where $n$ is the size of the
group). Further, they also show that $\CGM$ for nilpotent groups is in
the class $\AC((\log\log n)^2)$ and $\CGM$ for solvable groups of
class $d$ are in $\AC(d\log\log n)$. The interesting questions in the
static setting is to improve these upper bounds and/or extend these
results to other classes of groups.

In this paper we study the \emph{dynamic parallel complexity} of
$\CGM$ for abelian groups. Before we describe our results, we give
some background.

\paragraph*{\textbf{Dynamic complexity}}

Dynamic algorithms, broadly, deals with the design of efficient
algorithms for problems when the input is modified with small
changes. The aim is to solve the problem, for the modified input,
significantly more efficiently than running the best known ``static''
algorithm from scratch. The measure of efficiency is crucial here and
defines the model of computation. Dynamic algorithms is a burgeoning
field of research ( see e.g. \cite{HHS} and \cite{ILMP,DDKPSS,NO})
with many applications that require handling large inputs subject to
small changes over time.

From a \emph{parallel complexity perspective}, we have the framework
of Patnaik and Immerman~\cite{PatnaikI} that is rooted in descriptive
complexity \cite{Immerman}. Closely related is the work of Dong, Su,
and Topor~\cite{DongST95}. Here the ideal solution is to obtain a
dynamic algorithm for the considered problem in \emph{constant
  parallel time}. Theoretically, constant parallel time is $O(1)$ time
on a CRCW PRAM model (where the CRCW model is the most liberal as it
allows for concurrent reads and writes). It is well-known that this
coincides with the complexity class $\AC^0$ (the class of problems
solvable by constant-depth boolean circuits). From a descriptive
complexity perspective, when the circuits are dlogtime uniform (more
details in Section~\ref{sec2}) this corresponds to $\FO$, the class of
problems expressible in first-order logic. The dynamic complexity
class $\DynFO$ \cite{PatnaikI} is precisely the class of problems
solvable in $\DynFO$ (equivalently $\DynAC^0$), which means under
small changes the problem can be solved using a first-order
formula. These different ways of describing $O(1)$ parallel time are
essentially equivalent: because $\FO$ and \emph{uniform} $\AC^0$ are
equivalent \cite{BIS}. There is renewed interest in this model of
computation since a long-standing open problem, whether directed graph
reachability is in $\DynFO$, under single edge changes
\cite{PatnaikI}, was resolved in the affirmative~\cite{DKMSZ18}.

In the present paper, it is more convenient to describe our results,
which are essentially algorithmic and do not have a logical flavor, in
terms of the parallel circuit class $\AC^0$.

\paragraph*{\textbf{The results of this paper}}\hfill{~}

\vspace{2mm}

%\begin{comment}
In this paper, we obtain results on the dynamic parallel complexity of
abelian $\CGM$ and abelian Group Isomorphism. Our motivation is to see
if we can exploit the underlying group structure to give a $\DynAC^0$
algorithm for the $\CGM$ membership queries while the generating set
$S$ is dynamically changing with insertions and deletions. We are able
to obtain for the abelian group case the following results.
%\end{comment}

$\bullet$ First, we consider the more general problem of Monoid
Membership, where $G$ is a \emph{monoid} input by its multiplication
table. When $G$ is a \emph{commutative monoid} we give a deterministic
$\DynAC^0$ algorithm for membership testing that supports $O(1)$
insertions and deletions in each step. The main idea is to maintain
the monoid $M$ in a tree-like data structure. The cyclic monoids are
at the leaves of the tree and each internal node has the submonoid of
$M$ generated by set of all its descendant leaves. Furthermore, each
internal node will also hold submonoids corresponding to deletions of
its descendant nodes.
  
$\bullet$ We can use this tree-like data structure more powerfully in
the case of abelian groups to obtain a randomized $\DynAC^0$ algorithm
for abelian $\CGM$ that supports $\polylog(|G|)$ insertions/deletions
to $S$ in each step. The main fact that we exploit here is that adding
$\polylog(n)$ many unary numbers can be done in $\AC^0$. Thus, from an
abelian subgroup $H$ given by $\polylog(n)$ many generators we can
randomly sample from $H$ and hence list out all of $H$ with high
probability in $\AC^0$.

$\bullet$ If the number of insertions/deletions is at most $O(\log
n/\log\log n)$ then we obtain a deterministic $\DynAC^0$ algorithm for
abelian $\CGM$. Here our techniques are linear algebra based: we need
to consider some \emph{miniature} linear algebra problems: where the
number of variables is $O(\log n)$ and we adapt existing linear
algebraic techniques to solve this.

$\bullet$ We obtain analogous results for the dynamic abelian Group
Isomorphism.

\paragraph*{\textbf{The techniques used}}\hfill{~}

A brief word about the techniques we use in the paper. The main
dynamic complexity technique used is the idea of \emph{muddling}
\cite{DMSVZ19}. Intuitively, we divide the computation into time
intervals (of length $i(n)$ for size $n$ inputs and a suitable
function $i(\cdot)$). For a query that will arrive in time instant
$t$, the algorithm starts rebuilding the data structure that is
maintained at time instant $t-i(n)$. At any given instant, therefore,
the algorithm is running upto $i(n)$ threads of computation. This
broad technique is applied in this paper to the problems
considered. In addition, we will use some elementary group theory,
especially of finite abelian groups, and some tree-like data
structures.

\paragraph{\textbf{Organization}}

 In Section~\ref{sec2} we give some basic definitions and notation,
 and some background about the dynamic parallel complexity model. In
 Section~\ref{sec3} we explain the $\DynAC^0$ algorithm for
 commutative monoid membership under single insertions/deletions to
 the generating set. Sections~\ref{sec4} and \ref{sec4a} contain,
 respectively, the randomized and deterministic $\DynAC^0$ algorithms
 for abelian $\CGM$.  In Section~\ref{sec:AbIso} we apply the $\CGM$
 results to obtain $\DynAC^0$ algorithms for Abelian Group
 Isomorphism. Finally, in Section~\ref{sec6} we discuss dealing with
 small changes to the group multiplication table itself.

 \section{Preliminaries}\label{sec2}

 \paragraph*{\textbf{Complexity Classes}} We will mainly consider
 parallel complexity classes defined by boolean circuits. Let
 $\AC(t(n))$ denote the class of decision problems that have
 polynomial-size circuits of \emph{depth} $t(n)$ for inputs of size
 $n$, where the AND and OR gates of the circuit are allowed to be
 unbounded fanin. This circuit model is essentially equivalent to
 $t(n)$ parallel time on a CRCW PRAM model (with polynomially many
 processors), where CRCW allows for concurrent reads and writes to a
 memory location.  More details of these connections can be found in
 \cite{Immerman}. In particular, $\AC(1)$ is usually denoted $\AC^0$
 and $\AC^1$ denotes $\AC(\log n)$. The class $\AC(\log\log n)$ is of
 interest in this paper due to the result of Barrington et al
 \cite{BKLM} showing that abelian $\CGM$ is in uniform $\AC(\log\log
 n)$.  This class is also denoted $\FOLL$ in \cite{BKLM} (for
 first-order formulas with $\log\log n$ depth quantifiers for size $n$
 inputs). We use both notations interchangeably.

 An $\AC((t(n))$ algorithm will actually be given by a family of
 circuits $\{C_n\}_{n>0}$, where $C_n$ solves the problem for inputs
 of length $n$, has depth $t(n)$, and size bounded by some polynomial
 $n^c$ for a constant $c>0$. We need a \emph{uniformity condition}
 that tells us how efficiently we can construct the circuits $C_n$. A
 stringent condition is the so-called \emph{dlogtime uniformity}: Each
 gate in the circuit $C_n$ can be described using $O(\log n)$ bits and
 the uniformity condition requires that the gate connections can be
 checked in deterministic time linear in $O(\log n)$ by a random
 access machine \cite{BIS}. The class dlogtime uniform $\AC^0$
 coincides with $\FO$, where the structures on which the formulas are
 evaluated are equipped with some suitable predicates \cite{BIS}.

 The parallel dynamic algorithms in this paper are describable by
 circuits that are dlogtime uniform.
 
\paragraph*{\textbf{The parallel dynamic complexity model}} 

We briefly explain the parallel dynamic complexity model. The
underlying model for describing the algorithms can seen as a CRCW
PRAM.  That means the algorithm can use polynomially many parallel
processors accessing a shared memory that allows concurrent reads and
concurrent writes with well defined notion of which write succeeds. We
can also give a circuit complexity description for the model.

\begin{enumerate}
\item For each problem there is a well-defined notion of small
  changes to the input. 
\item In the CRCW PRAM setting, the algorithm uses $n^c$ processors
  for length $n$ inputs for some constant $c>0$ that depends on the
  problem.
\item At each time instant the algorithm receives as input $i(n)$
  small changes to the input. To the input at time instant $t$, the
  algorithm is required to output the answer in constant time. I.e.,
  within time instant $t+O(1)$.
\item In the boolean circuit setting, for inputs of length $n$ the
  model can be seen as a layered boolean circuit that is of width
  $n^c$ for length $n$ inputs, where the layers denote the time
  instants. At each layer it receives as input the $i(n)$ changes. For
  the input at layer $t$ it needs to output the answer before layer
  $t+O(1)$.

\item We use $\DynAC^0$ to broadly denote the class of problems that
  have dynamic algorithms that take $O(1)$ time with polynomially many
  processors.\footnote{This coincides with $\DynFO$ \cite{PatnaikI}
    when the dlogtime uniformity conditions are met.} We explicitly
  state the number of small changes to the input that can be handled
  at each time step. We also refer to such dynamic algorithms as
  $\DynAC^0$ algorithms.
\end{enumerate}

Depending on the problem at hand, the dynamic algorithm usually works
by creating a suitable data structure from the given input which it
updates with the small changes to the input.

\begin{comment}
The algorithm maintains a data structure $\DS$ which is recreated from
time to time. Additionally, the algorithm also maintains an additional
data structure $\ADS$ of size bounded by $t(n).b(n)$ which is the
number of changes to the input in a $t(n)$ time interval and $b(n)$ is
a bound on the number of changes at each time step.

The time line of computation is divided into intervals of length
$t(n)$, where $t(n)$ is the time required to rebuild the data
structure $\DS$ along with additional structure $\ADS$ into a new
version of $\DS$.

The algorithm, at each time instant $t$, starts rebuilding the data
structure $\DS(t)$ from the current input at time $t$. This rebuilding
of $\DS(t)$ is completed in $t(n)$ steps at time $t+t(n)-1$. At this
time $\ADS(t)$ (the bulk changes that arrived in the interval
$[t,t+t(n)-1]$). The algorithm answers any problem query at this time
step in $\C$.

\end{comment}

%\input{CGM1}

\section{A Dynamic $\CGM$ Algorithm for Commutative Monoids}\label{sec3}

In this section, we consider the more general problem of \emph{Cayley
  Monoid Membership} for commutative monoids: Given a commutative
monoid $M$ by its multiplication table, a subset $S\subseteq M$, and
an element $m\in M$, check if $m$ is in the submonoid $\angle{S}$
generated by $S$. By abuse of notation, we term this the $\CGM$
problem for commutative monoids.

We present a tree-based data structure to maintain the generating set
$S$, using which we obtain a $\DynAC^0$ algorithm that supports a
single insertion/deletion to/from the subset $S$ at each step. We will
use this data structure with suitable modifications in Section~\ref{sec4}.

\paragraph{\textbf{The $\CGM$ problem for monoids}}

It is known that $\CGM$ for monoids is reducible to directed graph
reachability. To see this just construct the Cayley digraph
$\Cay(M,S)$ of the monoid $M$ corresponding to generating set $S$. The
graph has vertex set $M$ and for every $m\in M$ and every $m' \in S$,
the directed edge $(m,mm')$ is in the edge set. Clearly, an element
$t\in M$ is in the submonoid $\ang{S}$ iff there is a directed path
from the monoid identity $e$ to $t$ in this digraph. Hence Cayley
Membership for monoids is in $\NL$.

We will use the weaker upper bound of $\LogCFL = \SAC^1$ for the
$\CGM$ problem for commutative monoids.\footnote{Recall that an
$\SAC^1$ circuit is a restricted form of $\AC^1$ circuit: it is of
logarithmic depth, is semi-unbounded, and is allowed negations only at
the input gates. That means, either each AND gate in the circuit has
fanin $2$ or, equivalently, each OR gate has fanin $2$ \cite{Venk}.}
  This upper bound actually gives a tree-like data structure, using
  which we obtain the $\DynAC^0$ algorithm for the problem that
  supports single insertions/deletions.

As $M$ is commutative, any element of the submonoid $\angle{S}$ is
expressible as a product of powers of elements in $S$: $\prod_{s\in
  S}s^{e_s}$, where $0 \le e_s\le n$ and $n=|M|$.

We claim that the entire submonoid $\angle{S}$ can be listed as the
output of an $\SAC^1$ circuit. The circuit takes $S$ as input, as a
$n$-bit binary number with $i^{th}$ bit indicating whether the element
$m_i\in M$ is in $S$, and outputs an $n$-bit binary number whose $i$th
bit is $1$ iff $m_i\in M$ is in the submonoid $\angle{S}$.

Let $S_1,S_2\subseteq M$ such that the monoid identity $1$ is in both
$S_1$ and $S_2$. Their product is $S_1S_2=\{ab\mid a\in S_1, b\in
S_2\}$.

\begin{proposition}\label{prop:setMult}
Given as input the subsets $S_1$ and $S_2$ of a monoid $M$ their
product $S_1S_2$ can be computed in $\SAC^0$.
\end{proposition}

For each $a\in S$ let $P_a = \{a^i: 1 \leq i \leq n\}$. Notice that
$P_a$ can be computed directly from the multiplication table for $M$
in $\L$ which is contained in $\SAC^1$.

\paragraph{\textbf{The tree-like data structure for $S$}}

We create a balanced binary tree $T$ with leaves labeled by the
distinct elements $a\in S$. To the leaf labeled $a$ we associate the
subset $P_a$. Inductively, to each internal node $u$ of $T$ we
associate the product $M_u=M_vM_w$, where $v$ and $w$ are its two
children. Since $|M|=n$, the tree $T$ has depth bounded by $\log
n$. By Proposition~\ref{prop:setMult} the tree $T$ can be created by
an $\SAC^1$ circuit. 

\begin{proposition}\label{tree-monoid-static}
  For each node $u$ of the tree $T$ the subset $M_u$ associated with
  $u$ is the submonoid generated by the subset $L_u=\{a\in S\mid a$
  such that $u$ has the leaf labeled $a$ as descendant$\}$. I.e.,
  $M_u=\angle{L_u}$.
\end{proposition}

\begin{proof}
  This is easily proved by induction on the tree $T$. The point to
  note is that commutativity of the monoid $M$ is crucial: if
  $M_v=\angle{L_v}$ and $M_w=\angle{L_w}$ by induction hypothesis,
  then we note that $M_vM_w=\angle{L_v\cup L_w}$ because all the
  elements commute with each other. Hence $M_u=M_vM_w=\angle{L_u}$.
\end{proof}

\begin{lemma}
  The $\CGM$ problem for commutative monoids is in $\SAC^1$.
\end{lemma}

\begin{proof}
  Given $m\in M$, to check if $m\in\angle{S}$ we just need to
  check if $m$ is in the submonoid $S_r$ labeling the root $r$
  of $T$. 
\end{proof}

\paragraph{\textbf{The Dynamic Setting}}

We will modify the above construction to obtain a dynamic data
structure. First we expand the tree by also including leaves for each
element of $M\setminus S$. For each leaf $a\in M$ we pre-compute its
power set $P_a$ as already defined. However, we will associate $P_a$
to the leaf labeled $a$ precisely if $a\in S$, and otherwise we
associate the identity element $\{1\}$ with the leaf $a$.

We will need to dynamically maintain the following data at each node
of the expanded tree.

\begin{enumerate}
%\item $\anc(\nu,\nu')$ This is a convenience predicate that does not change
%and is initialised to true iff $\nu$ is an ancestor of $\nu'$.
%Notice that since $T$ is a full binary tree, the parent of a node
%$\nu$ is $\floor{\nu/2}$.
%The $\anc$ relation is the transitive closure of the parent relation.
%This allows us to encapsulate the powers of $\nu$ in the same tree rather
%than precomputing them.

\item For each node $\nu$ of the tree we have the submonoid $M_\nu$
  generated by the subsets associated with the leaves below node $\nu$
  in the tree.

\item Additionally, at each node $\nu$ we will maintain the submonoid
  $M_\nu:\M_\mu$ for each descendant $\mu$ of $\nu$ in the tree, where
  $M_\nu:M_\mu$ denotes the submonoid generated by the subsets
  associated with all leaves that are descendants of $\nu$ \emph{but
    are not} descendants of $\mu$. Equivalently, it is as if the
  submonoid associated with node $\mu$ is reset to $\{1\}$ and then
  the rest of the submonoid $M_\nu$ is computed.
\end{enumerate}

Essentially, as in Propositions~\ref{prop:setMult} and
\ref{tree-monoid-static}, given a subset $S$ we can construct the tree
data structure along with the data at each node as described above.

\begin{proposition}\label{tree-monoid}
  Given as input a subset $S\subseteq M$ for a commutative monoid $M$
  (given by its multiplication table), we can construct the tree data
  structure along with the data at each node as described above in
  $\SAC^1$. Furthermore, given a membership query $m\in M$, testing
  if $m$ is in the current submonoid $\angle{S}$ can be done in $\AC^0$.
\end{proposition}

\begin{proof}
 For the tree construction, it suffices to observe that we can do the
 computation of each $M_\nu:M_\mu$ in parallel in $\SAC^1$. For
 membership testing, if the $\rho$ is the root of the tree then the
 submonoid $M_\rho = \angle{S}$ is available as a list at the node
 $\rho$. Hence membership testing is in $\AC^0$.
\end{proof}

%We have the following from the definition of $\nodeSet$:
%\begin{proposition}\label{prop:cmm}
%$\nodeSet(r,t)$ is true iff $t$ is generated by $S$ where $r$ is
%the root node of the tree $T$.
%\end{proposition}

\paragraph{Handling single insertions and deletions}

Next we show that the above data structure supports single insertions
and deletions to $S$ at each time step. The updates to the data
structure can be carried out in $\AC^0$ as described below.

First we consider deletions. Suppose $a\in S$ is deleted. Then
the following changes are carried out.
\begin{enumerate}
\item The set associated to the leaf node corresponding to $a$ is
  reset to $\{1\}$.
\item For each tree node $\nu$, if the leaf node $a$ is a descendant
  of $\nu$ then the submonoid $M_\nu$ is replaced with $M_\nu:P_a$,
  which is already pre-computed at node $\nu$. If $a$ is not a
  descendant of $\nu$ then no changes are required at node $\nu$.
\item For each tree node $\nu$ and each descendant $\mu$ of $\nu$
  consider the submonoid $M_\nu:M_\mu$. If if $a$ is a descendant of
  $\mu$ then $M_\nu:M_\mu$ needs no changes. Now suppose $a$ is a
  descendant of $\nu$ but not a descendant of $\mu$. Let the least
  common ancestor of $a$ and $\mu$ be $\mu'$. Let $\mu_1$ and $\mu_2$
  be the children of $\mu'$ such that $a$ is a descendant of $\mu_1$
  and $\mu$ is a descendant of $\mu_2$. Then $M_\nu:M_\mu$ is
  recomputed as the product $(M_\nu:M_{\mu'})\cdot
  (M_{\mu_1}:P_a)\cdot (M_{\mu_2}:M_{\mu})$.
\end{enumerate}

As each node in the tree can be processed in parallel we have
the following.
  
\begin{proposition}
  The above deletion operations can be carried out in $\AC^0$.
\end{proposition}

Next we consider insertions. Let $a\in M$ be inserted in $S$.

\begin{enumerate}
\item At the leaf node $a$ the associated submonoid is set to $P_a$.
\item For each node $\nu$ containing the leaf node $a$ as descendant,
  we can update $M_\nu$ directly to the product submonoid $M_\nu\cdot
  P_a$.
\item For each node $\nu$ and a descendant $\mu$ we will update
  $M_\nu:M_\mu$ as follows: if the leaf node $a$ is not a descendant
  of $\nu$ then no changes are required. If it is a descendant of
  $\mu$ as well then also no changes are required. Otherwise, let
  $\mu'$ be the least common ancestor of $a$ and $\mu$, and the two
  children of $\mu'$ be $\mu_1$ (containing $a$ as leaf node) and
  $\mu_2$ (containing $\mu$ as descendant). Then we replace the
  submonoid $M_\nu:M_\mu$ with the product submonoid
  $(M_\nu:M_{\mu'})\cdot (M_{\mu_1}\cdot P_a)\cdot (M_{\mu_2}:M_\mu)$.
\end{enumerate}

These updates can all be carried out in parallel and each requires just
a product of a constant number of pre-computed submonoids. Hence
we have the following

\begin{proposition}\label{insertions=monoid}
The above insertion operations can be carried out in $\AC^0$.
\end{proposition}

To summarize the above results, we have the main theorem of this
section.

\begin{theorem}\label{thm:monoid}
  Let $M$ be a commutative monoid given by its multiplication table
  and $S\subseteq M$ generate a submonoid $\angle{S}$ of $M$. Then
  there is a deterministic $\DynAC^0$ algorithm that answers
  membership queries $m\in\angle{S}$ given $m\in M$ and supports
  single insertions and deletions to the generating set $S$ at each
  time step.
\end{theorem}

%Since by Lemma~\ref{lem:monoidRecons} we know that there is at most one 
%monoid that the given magma can turn into in the next $O(\log{n})$ steps,
%we can, by using the muddling lemma (Lemma~\ref{lem:muddle}) restrict our
%attention to a single monoid for $\log{n}$ steps. If this monoid happens
%to be commutative w

%We use Lemma~\ref{lem:DynNSCond} and Proposition~\ref{prop:cmm}
%maintain the Cayley membership problem for commutative monoids for
%varying $S$ values.  Thus, we have completed the proof of the main
%theorem.

%\input{abel-dynamic}

%\maketitle
%\thispagestyle{empty}

\section{The Dynamic $\CGM$ Problem for Abelian Groups}\label{sec4}

Let $G$ be an $n$-element abelian group given by its multiplication
table. Let $n=p_1^{a_1}\times p_2^{a_2}\times \cdots \times
p_\mu^{a_\mu}$ be its prime factorization, where $p_i$ are distinct
primes. By the structure of finite abelian groups \cite[Theorem
  3.3.1]{Hall} $G = G_1\times G_2\times \cdots G_\mu$ is a direct
product where $G_i$ is the $p_i$-Sylow subgroup of $G$.\footnote{Let
$G$ be a finite group of order $n=p_1^{a_1}\times p_2^{a_2}\times
\cdots \times p_\mu^{a_\mu}$. Then for each $i$, $G$ has at least one
subgroup of order $p_i^{a_i}$, which is known as a $p_i$-Sylow
subgroup of $G$. However, if $G$ is abelian then there is a unique
$p_i$-Sylow subgroup of $G$ which we can denote by $G_i$. In that
case, $G = G_1\times G_2\times \cdots G_\mu$.}

For a subset $S\subseteq G$ consider the subgroup $H=\angle{S}$
generated by $S$. Let $b_i=n/p_i^{a_i}, 1\le i\le \mu$ and $S_i =
\{x^{b_i}\mid x\in S\}, ~~1\le i\le \mu$.  Then $H_i=\angle{S_i}$ is a
subgroup of $G_i$ for each $i$ and
\[
H = H_1\times H_2\times \cdots \times H_\mu.
\]

An element $g\in G$ is in the subgroup $H$ if and only if $g^{b_i}\in
H_i$ for each $i$. As a consequence we can reduce the dynamic abelian
$\CGM$ problem to the dynamic abelian $\CGM$ problem for abelian
$p$-groups. We state this as a lemma.

\begin{lemma}
  Given an $n$-element abelian group $G$ by its Cayley table, by a
  one-time preprocessing computation in $\AC^1$ (or even polynomial
  time suffices for this purpose) we can compute Cayley tables for
  each Sylow subgroup $G_i$. Furthermore, all powers of elements of
  $G$ can be pre-computed and stored in an array. Hence the parallel
  dynamic complexity of abelian $\CGM$ maintaining $S$, supporting say
  $t(n)$ insertions/deletions at each step is $\AC^0$ reducible to the
  same problem for abelian $p$-groups.
\end{lemma}  

\paragraph{\textbf{Dynamic $\CGM$ for abelian $p$-groups}}

Let $G$ be an $n$-element abelian $p$-group given as input by its
multiplication table, where $n=p^m$. Let $S=\{g_1,g_2,\ldots,g_s\}$
be a subset of $G$. We want to maintain a structure that supports
efficient membership testing in the subgroup $H=\angle{S}$ generated
by $S$. That means efficiently supporting the following operations.

\begin{enumerate}
\item Given a query element $g\in G$ test if $g\in H=\angle{S}$ and,
  if so, express $g$ as a product of the generators in $S$.
\item The dynamic version of the model requires that we efficiently
  support insertions/deletions to set $S$. We have seen how to handle,
  even in the setting of commutative monoids, single insertions or
  deletions to $S$ at each step. More generally, we would like to
  handle bulk insertions and deletions at each step.
\end{enumerate}  

\paragraph{\textbf{Preprocessing for dynamic abelian $\CGM$}}

For the dynamic algorithm, we will first preprocess the abelian
$p$-group $G$ given by multiplication table. Let $|G|=p^m=n$. We will
first obtain a generating set of size at most $\log n$ for $G$ using
which we can represent all elements of $G$.

Let $G$ be a finite abelian group. An \emph{independent} generating
set \cite[Theorem 3.2.2]{Hall} for $G$ is a generating set
$\{g_1,g_2,\ldots,g_\ell\}$ such that $g_1^{e_1}\cdot g_2^{e_2}\cdots
g_\ell^{e_\ell}=1$ for $0\le e_i\le o(g_i)-1$ if and only if $e_i=0$,
where $o(g_i)$ is the order of $g_i$ for each $i$. As a consequence of
independence, every element $g\in G$ is \emph{uniquely expressible} as
$g=g_1^{e_1}\cdot g_2^{e_2}\cdots g_\ell^{e_\ell}$ for $0\le e_i\le
o(g_i)-1$. It is easy to see that $\ell\le \log|G|$. The next
proposition easily follows from \cite[Theorem 3.2.2]{Hall}.

  \begin{proposition}\label{ind-gen-find}
    Given as input a finite abelian group $G$ by its multiplication
    table, an independent generating set for $G$ (of size at most
    $\log |G|$) can be computed in $\AC(\log n)$.
   \end{proposition} 

 Thus, in a one-time preprocessing step, we can compute an independent
 generating set $\{g_1,g_2,\ldots,g_\ell\}$ from the multiplication
 table of the input group $G$ as well as the unique expression for
 each $g\in G$ as $g=g_1^{e_1}\cdot g_2^{e_2}\cdots g_\ell^{e_\ell}$
 for $0\le e_i\le o(g_i)-1$. The two preprocessing steps for $G$ are
 summarized below.

\begin{enumerate}
\item For each pair $(g,i), g\in G, 0\le i\le n-1$, we compute and
  store the power $g^i$ in an $n\times n$ table. The product of two
  elements in the table is computable in $\AC^0$. This is a
  straightforward $\AC(\log n)$ computation. In particular, this
  computation also yields the order $o(g)$ of each $g\in G$.
\item By Proposition~\ref{ind-gen-find}, in $\AC(\log n)$ we compute
  for $G$ an independent generating set $T=\{g_1,g_2,\ldots,g_t\}$
  $t\le \log n$ and also a representation for each $g\in G$ as a
  product
  \[
  g=\prod_{i=1}^t g_i^{e_{gi}}.
  \]
\end{enumerate}

\begin{comment}
\begin{remark}
  We can actually do these preprocessing steps in $\AC(\log\log n)$
  using the Barrington et al strategy \cite{BKLM}. However, since the
  preprocessing is done only once, we do not require this more
  involved algorithm.
\end{remark}
\end{comment}

Additionally, we note the easy consequence of Barrington et al's
$\FOLL$ algorithm \cite{BKLM} for abelian $\CGM$.

\begin{lemma}\label{preproc}
  Let $S\subseteq G$, for an abelian group $G$ given by its
  multiplication table as input. In $\FOLL$ a $\log |G|$ size subset $T$
  of $S$ can be computed that generates the same subgroup as $S$.
\end{lemma}

\begin{proof}
  Let $S=\{x_1,x_2,\ldots,x_s\}$. For each $i>1$ in parallel, we can
  check if $x_{i+1}$ is in $\angle{x_1,x_2,\ldots,x_i}$ using the
  $\FOLL$ algorithm of \cite{BKLM}. If $x_{i+1}\notin
  \angle{x_1,x_2,\ldots,x_i}$ then we include $x_{i+1}$ into the set
  $T$. Clearly, $|T|\le \log |G|$ and generates the same subgroup and
  $S$.
\end{proof}

\subsection{Randomized $\DynAC^0$ Algorithm for Abelian $\CGM$}

We first present a randomized $\DynAC^0$ algorithm for maintaining the
subgroup $H=\angle{S}$ of $G$, given by its generating set. More
precisely, the algorithm can process $\polylog(n)$ insertions and
deletions in each step and answer membership queries to the group
$\angle{S}$ in $O(1)$ parallel time (i.e. in $\AC^0$).
  
The next observation is key to the algorithm in this section.

\begin{lemma}\label{randac0}
  Let $T\subset G$ be of size at most $\log^cn$ for some constant $c$,
  where $|G|=n$ is an abelian group given by its multiplication table
  with independent generating set $G=\angle{g_1,g_2,\ldots,g_\ell}$.
  Then in randomized $\AC^0$ we can list out the subgroup $\angle{T}$
  generated by $T$. In particular, membership testing in the subgroup
  $T$ can be done in randomized $\AC^0$.
\end{lemma}

\begin{proof}
  Let $T=\{x_1,x_2,\ldots,x_r\}$. For each $g\in G$ we have the
  pre-computed unique product
  \[
  g=\prod_{i=1}^\ell g_i^{\alpha_i},
  \]
  using the independent generating set $\{g_1,g_2,\ldots,g_\ell\}$.
  In particular, for each $x_j\in T$ we have
  \[
  x_j=\prod_{i=1}^\ell g_i^{\alpha_{ij}},
  \]
  where $0\le \alpha_{ij} \le o(g_i)-1$ for each $i\in[\ell]$. As
  explained below, we can randomly sample from the group generated by
  $T$ by picking numbers $\beta_j\in_R[n], 1\le j\le r$ uniformly at
  random and computing the product
  \[
  x = \prod_{j=1}^rx_j^{\beta_j}.
  \]
  The number of such products is $n^r$. Furthermore, each element of
  the subgroup $\angle{T}$ occurs in this product with multiplicity
  exactly $|\{(\beta_1,\beta_2,\ldots,\beta_r)\mid \prod_j
  x_j^{\beta_j}=1\}$, as this set is the kernel of the group
  homomorphism mapping $(\beta_1,\beta_2,\ldots,\beta_r)\mapsto
  \prod_{j=1}^rx_j^{\beta_j}$. Thus, $x$ is uniformly distributed in
  $\angle{T}$. If we draw, say $n^2$ such samples $x$ in parallel, the
  probability that all elements of $\angle{T}$ appear is at least
  $1-e^{-n}$. Finally, we analyze the complexity of computing the
  product $x=\prod_{j=1}^rx_j^{\beta_j}$. Notice that it amounts to
  computing the product $\prod_{i=1}^\ell g_i^{\sum_{j=1}^r\beta_j
    \alpha_{ij}}$. Now, each of these $\ell$ exponents
  ${\sum_{j=1}^r\beta_j \alpha_{ij}}$ is a $\log^c n$ sum of
  \emph{unary} numbers and can be computed modulo the unary number
  $o(g_i)$ in $\AC^0$. The final product can be looked up in the
  pre-computed table to find $x$. This proves that the group
  $\angle{T}$ can be listed in randomized $\AC^0$ and hence membership
  testing in $\angle{T}$ is also in randomized $\AC^0$.
\end{proof}

\paragraph{\textbf{Statically created data structure}}

First we will create the tree-based data structure to represent the
generating set $S$ and certain subgroups of $\angle{S}$. This data
structure we will continually update as we process the bulk insertions
and deletions that occur at each time step.

\begin{enumerate}
\item Let $S$ be the current generating set and $2^{k-1}<|S|\le 2^k$,
  for positive integer $k\le \lceil \log n\rceil$. To begin with, we create an
  $O(\log n)$ depth \emph{full binary tree} with $2^k$ leaves.  The
  tree will have $2^k$ leaves. There are $|S|$ leaves, one for each
  generator $x\in S$. The corresponding leaf of the tree is labeled by
  the \emph{cyclic subgroup} $\angle{x}$ generated by $x$. The
  remaining leaves are labeled by the trivial subgroup $1$. As it is a
  full binary tree, its nodes can be indexed by $1,2,\ldots,2^{k+1}-1$
  with $1$ as index for the root. For each $i>1$, the node indexed $i$
  has as parent the node indexed $\lfloor i/2 \rfloor$.
  
\item Let $v$ be an internal node of the tree with children $u$ and
  $w$. We inductively compute at $v$ the subgroup $H_v=H_u\cdot H_w$
  which is the product of the subgroups $H_u$ and $H_w$ computed at
  the two children. Notice that the product $H_uH_w$ is indeed a
  subgroup of $G$ as $G$ is abelian. Letting $S_v$ denote the set of
  leaves below node $v$. Then, notice that $H_v=\angle{S_v}$ for each
  node $v$ of the tree. The root is labeled with $H=\angle{S}$.
\item Additionally, for each internal node $v$ and for each
  \emph{descendant} $u$ of $v$ we keep the subgroup generated by
  $S_v\setminus S_u$, which we denote by $H_v:H_u$. Thus, at each node
  $v$ we have the list of subgroups $H_v:H_u$ with generating set
  $S_v\setminus S_u$, one for each descendant $u$ of $v$.
\item Finally, using Lemma~\ref{preproc} we compute, in $\FOLL$, $\log
 n$ size generating sets $T_{vu}\subset S_v\setminus S_u$ for each
  subgroup $H_v:H_u$ at each node $v$ in parallel.

\item For access to the data maintained at each node in the tree, we
  will have an array of pointers indexed by $1,\ldots,2^{k+1}-1$.
  Furthermore, we will keep a boolean array $A[i,j], 1\le i,j\le
  2^{k+1}-1$ where $A[i,j]=1$ if and only if $i$ is an ancestor of
  $j$. 
  
\end{enumerate}

The following lemma is immediate.

\begin{lemma}
  The data structure for $S$ as described above can be built in $O(\log n)$
  parallel time (i.e. in $\AC(\log n)$). 
\end{lemma}

\begin{proof}
  The tree has $\log n$ levels and the straightforward computation
  required at each of the (at most $|S|$) nodes at each level is
  $\AC^0$.
\end{proof}  

\paragraph{\textbf{Handling bulk insertions and deletions}}

At any point of time during the computation, the current generating
set is maintained as a data structure $S$, described above, along with
two sets $I$ and $D$ such that $|I|,|D|=O(\log^{c+1}n)$, where the
actual generating set is $S\cup I\setminus D$. We will first show that
a membership query occurring at this point of time can be answered in
$O(1)$ parallel time.

\begin{lemma}\label{polylog1}
  Given the data structure for $S$ along with the update sets of
  insertions $I$ and deletions $D$, each of $\log^{c+1}n$ size, we can
  test if some group element $g$ is in the subgroup generated by
  $S\cup I\setminus D$ in $\AC^0$.
\end{lemma}

\begin{proof}
Let $D=\{x_{i_1},x_{i_2},\ldots,x_{i_d}\}$ be the deletions from $S$.
From the data structure for $S$ we can find the $d$ subtrees rooted at
nodes $v_{i_1},v_{i_2},\ldots,v_{i_d}$ where each node $v_{i_j}$ is
root of the maximal subtree that has exactly the one deletion
$x_{i_j}$ occurring among its leaves. This can be done in $O(1)$
parallel time using the ancestor boolean array $A[u,v]$. In the data
structure for $S$ we already have a $\log n$ size generating set, say
$T_j$, for each subgroup $H_{v_{i_j}}:\angle{x_{i_j}}$.

Additionally, we find the \emph{maximal subtrees}, rooted at nodes
$u_{\ell_1},u_{\ell_2},\ldots,u_{\ell_s}$ of the tree, such that these subtrees
contain no deletion $x_{i_j}$ as descendant. Each such maximal subtree
has as sibling a subtree that contains one or more deletions among its
leaves. The roots of all subtrees that have deletions among its leaves
are just the ancestors of the $v_{i_j}$ nodes, and hence are at most
$d\lceil\log n\rceil$ in number (each leaf has $k$ ancestors). Thus,
$s=O(d\log n)=O(\log ^{c+2}n)$.  At each node $u_{\ell_k}$ we already
have a $\log n$ size generating set, say $\hat{T}_k$ for each subgroup
$H_{u_{\ell_k}}$. Let $T=\cup_j T_j$ and $\hat{T}=\cup_k \hat{T}_k$. It
follows from the above that we can compute $T$ and $\hat{T}$ in $O(1)$
parallel time.

Putting it all together, the group generated by $S\cup I\setminus D$
is actually generated by $\hat{T}\cup T\cup I$ which is of
$\polylog(n)$ size. Therefore, applying Lemma~\ref{randac0} we can do
membership testing in this subgroup in randomized $\AC^0$.
\end{proof}

\paragraph{\textbf{Continual rebuilding of the data structure for $S$}}

This is the crucial part of the dynamic algorithm.\footnote{This
continual rebuilding of data structure is essentially the muddling
technique introduced in \cite{DMSVZ19}.}  Lemma~\ref{polylog1} shows
that we can handle membership queries in $O(1)$ parallel time using
the data structure for $S$ provided the sets $I$ and $D$ are
$\polylog(n)$ size bounded. Specifically, suppose at each time instant
there are $\log^cn$ many insertions/deletions. Then in $\log n$ time
steps, the sets $I$ and $D$ grow to size bounded by $\log^{c+1}n$. At
this point we can rebuild the static data structure in $\AC(\log n)$,
i.e. $\log n$ parallel time for the current generating set $S_1=S\cup
I\setminus D$. Crucially, during this $\log n$ time window fresh bulk
insertions/deletions, say $I_1$ and $D_1$, will arrive. Therefore, at
the end of the $\log n$ time window we have the data structure for
$S_1$ along with the $\log^{c+1}n$ size subsets $I_1$ and $D_1$ and we
can answer membership queries in $O(1)$ time using
Lemma~\ref{polylog1}.

To summarize, we have shown the following theorem.

\begin{theorem}\label{thm1}
  There is a randomized $\DynAC^0$ algorithm for the abelian Cayley
  Group Membership problem, $\CGM$, that supports $\polylog(n)$
  insertions and deletions to the generating set.
\end{theorem}

\section{A Deterministic Dynamic Algorithm for Abelian $\CGM$}\label{sec4a}

We now present a \emph{deterministic} $\DynAC^0$ algorithm for abelian
$\CGM$ that can process bulk insertions/deletions of size $t =
O({\frac{\log n}{\log\log n}})$. The algorithm is linear
algebraic. Recall that we have pre-computed an independent generating
set $\{g_1,g_2,\ldots,g_\ell\}$ for the abelian $p$-group $G$.  Let
$|G|=n=p^m$ and $o(g_i)=p^{m_i}$ for each $i$, where $m_1+m_2+\cdots +
m_\ell=m$, and $\ell\le \log n$. We also have $m=\log_p n\le \log
n$. Each $g\in G$ has a unique representation (pre-computed) as
\[
g=\prod_{i=1}^\ell g_i^{b_i},~~~ 0\le b_i\le p^{m_i}-1.
\]
Thus $g$ can also be represented as an $\ell$-dimensional integer
column vector $\bar{b}$.

We will dynamically maintain a subset $T$ of the generating set $S$
such that $|T|\le \log n$ and $\angle{S}=\angle{T}$.  As explained, we
will represent elements of $T$ $\ell$-dimensional column
vectors. Thus, $g\in \angle{T}$ iff the system of integer linear
equations
\[
Ax=\bar{b}
\]
is feasible, where the matrix $A$ has columns corresponding to each
generator in $T$, and the $i^{th}$ row of the system of equations is
computed modulo $p^{m_i}$ for $1\le i\le \ell$. We can suitably scale
each equation to get a system of integer linear equations modulo
$p^m$. Since $p^m$ is a composite for $m>1$, the usual recipe for
feasibility of linear equations based on matrix rank does not directly
apply. However, using some basic linear algebra we have the following.
We can rewrite $Ax=\bar{b}~(\Mod p^m)$ equivalently as integer linear
equations $Ax+p^my=\bar{b}$, where $y$ is a column vector of $\ell$
new variables. Letting $[A|p^mI_\ell]=\tilde{A}$, and $z=[x|y]^T$ we
can write this as $\tilde{A}z=\bar{b}$, noting that $\tilde{A}$ is
full row rank. Then we have

\begin{lemma}{\rm\cite[Theorem 3.13]{AV09}}\label{av-lem}
  For a prime $p$, the system $\tilde{A}z=\bar{b}$ (and hence
  $Ax=\bar{b}~(\Mod p^m)$) is feasible iff GCD of the $\ell\times
  \ell$ subdeterminants of $\tilde{A}$ and the GCD of the $\ell\times
  \ell$ subdeterminants of the augmented matrix $[\tilde{A}|\bar{b}]$
  have the same highest power of $p$ dividing them both.
\end{lemma}

We will also require the following lemma.

\begin{lemma}\label{lem:detPolylogInFOLL}
Let $A$ be a square matrix of dimension $\polylog(n)$ with entries
that are polynomially bounded in $n$, then $\det(A)$ can be computed
in $\AC(\log\log n)$.
\end{lemma}

\begin{proof}
It is well known that the determinant of a matrix of $n$ variables can
be computed by boolean threshold circuits\footnote{Threshold circuits
  allow for unbounded fanin threshold gates apart from NOT, AND, and
  OR gates.} of polynomial size and logarithmic depth, i.e.\ it is in
$\TC^1$. (Proof sketch: the determinant of a matrix of polynomial
dimension with polynomial in $n$ bit entries can be computed in
arithmetic $\SAC^1$ \cite[Table 2]{MV}. In other words, it can be
computed by a layered logarithmic depth circuit with gates from
$\{+,-,*\}$ where the $*$-gates have fan-in $2$. Now by applying
\cite{HAB} each layer of this arithmetic circuit can be simulated in
$\TC^0$, i.e.\ constant-depth threshold circuits). Hence, replacing
$n$ with $\log n$, it follows that the determinant of matrices of
$\polylog(n)$ dimension with $\polylog(n)$ bit entries can be computed
by a threshold circuit of depth $O(\log\log n)$ and size
$\polylog(n)$. Furthermore, threshold gates of $\polylog n$ fanin can
be computed by $\poly(n)$ size uniform $\AC^0$ \cite{Ajtai,Ajtai2}.
Replacing the threshold gates by the corresponding $\AC^0$-circuit of
size $\poly(n)$ completes the proof.
\end{proof}

\begin{comment}
\paragraph{\textbf{Static Data Structure for matrix $[A ~|~ p^mI]$}}

The integer coefficient matrix $\tilde{A}=[A~|~p^mI]$ is $O(\log
n)\times O(\log n)$.  Thus, it has only polynomially many square
submatrices. For all these square submatrices of $[A~|~p^mI$ we will
  precompute, in parallel, their inverses and determinants. By
  Lemma~\ref{lem:detPolylogInFOLL} this entire computation can be
  carried out in $\AC(\log\log n)$ because the matrices are $O(\log
  n)$-dimensional.
\end{comment}  

%We summarize the discussion below.

\begin{lemma}\label{samir1}
Let $Ax=\bar{b}~(\Mod p^m)$ be a system of integer linear equations
modulo $p^m$, where $p^m=n$ is input in unary, and
$A\in\mathbb{Z}^{t\times t}$ and $\bar{b}\in \mathbb{Z}^t$, $t=O(\log
n)$. Then we have the following:
\begin{enumerate}
\item Let $\tilde{A}=[A|p^m I_\ell]$. We can compute the determinants
  of all square submatrices of $\tilde{A}$ and $[\tilde{A}|\bar{b}]$
  in $\AC(\log\log n)$ (i.e. in $\log\log n$ parallel time).
\item Furthermore, for the nonsingular submatrices we can also compute
  their inverses in $\AC(\log\log n)$.
\item Given the above data we can test the feasibility of
  $Ax=\bar{b}~(\Mod p^m)$ and solve for $x$ in $\AC^0$.
\end{enumerate}  
\end{lemma}

\begin{proof}
For the first part, the number of square submatrices is polynomially
bounded as $\tilde{A}$ has dimension $O(\log{n}) \times
O(\log{n})$. Reducing modulo $p^m$, the entries of the matrix are
bounded by $n$. Thus, by Lemma~\ref{lem:detPolylogInFOLL} it follows
that the determinant as an integer can be computed in
$\AC(\log{\log{n}})$. Reducing modulo $p^m$ yields the answer and we
know from \cite{HAB} that the division by a unary number is possible
in $\AC^0$.

%% The value of the determinant
%% of each square submatrix $N$ viewed as a determinant over integers (where we
%% replace an entry of $N$ by the smallest integer less that $p^m$ 
%% to which it is congruent modulo $p^m$) is bounded by 
%% $(p^m-1)^{\log{n}}{(\log{n})!} = 2^{O(\log^2{n})}$. 
%% 
%% Thus by the prime number
%% the $O(\log^2{n})$ primes are enough to Chinese remainder this value.
%% We pick the first $O(\log^2{n})$ primes $q$ which are greater than $3\log{n}$. 
%% Now we can apply Lemma~\ref{lem:schur} to compute the
%% determinant of every square submatrix $N$ modulo $q$. We can now Chinese
%% remainder the determinants to get an integer value, which we reduce modulo $p^m$
%% to get the desired determinant of $N$ modulo $p^m$.

For the second part, consider every nonsingular submatrix $N$,
i.e.\ $\det{N} \bmod{p}$ is non-zero. We can compute the entries of
$N^{-1}$ modulo $p$ by Cramer's rule, as each cofactor of $N$ is also
a submatrix of $\tilde{A}$. Since we can compute division of
$O(\log{n})$-bit integers in $\AC^0$ (see \cite[Theorem 5.1]{HAB}) it
follows that these computations can also be done in
$\AC(\log\log{n})$.

For the last part, notice that feasibility can be tested in $\AC^0$,
given the data of first two parts, by Lemma~\ref{av-lem}. Let a
maximum dimension submatrix with non-zero determinant be
$N$. W.l.o.g.\ $N$ is the top left submatrix of $\tilde{A}$ and has
dimension $L$. Let $Nz' = \bar{b'}$ be the system of equations
obtained by truncating all rows below the $L^{th}$ from $\bar{b}$. The
solution to this is $z' = N^{-1}\bar{b'}$, where the right hand size
is computable in $\AC^0$ given $N^{-1}$. Now we can extend this
solution to a solution of $\tilde{A}z = \bar{b}$ by putting $z_i =
z'_i$ for each column $i$ in $N$ and $z_i = 0$ for other columns.  It
is easy to see that this must be a solution of $\tilde{A}z = \bar{b}$,
as we know the latter is feasible.
\end{proof}

This preprocessing of the matrix $\tilde{A}$ needs to be combined with
a variant of the matrix inverse lemma stated below \cite{HeSe81} (this
is a variant of the so-called Sherman-Morrison-Woodbury formula) to
dynamically compute solutions to $\tilde{A}z=\bar{b}$. This formula
essentially allows for a quick updation of the data computed using
Lemma~\ref{samir1} for $A$, if $A$ is replaced with $A+A'$ for a small
rank matrix $A'$.

\begin{lemma}[Binomial Matrix Theorem]\label{SMW}{\rm\cite{HeSe81}}
  Let $M$ be an invertible $r\times r$ matrix over any field (or
  ring).  Let $C$, $U$ and $V$ be $t\times t$, $r\times t$ and
  $t\times r$ matrices respectively, over the same field/ring. 
  If $M+UCV$ is invertible then the inverse is
  \begin{equation}\label{mil-eqn}
  (M+UCV)^{-1} = M^{-1} - M^{-1}U(I+CVM^{-1}U)^{-1} CVM^{-1}.
  \end{equation}
\end{lemma}

%Notice that the product $UCV$ is a $r\times r$ rank $t$ matrix.

Similarly to update determinants quickly we will need the following.

\begin{lemma}[matrix determinant lemma]\label{mdet-lem}
  If $M$ is a $r\times r$ matrix over a field and $U$ and $V$ are
  $r\times t$ and $t\times r$ matrices then
  \begin{equation}\label{mdl-eqn}
  \det(M+UV^T) = \det(I_t+V^TM^{-1}U)\det(M).
  \end{equation}
\end{lemma}

We will now see how to use these in the context of processing $O(\log
n/\log\log n)$ bulk insertions and deletions. Finally, we will also
require the following lemma to put everything together.

\begin{lemma}{\rm\cite[Theorem 8]{DMVZ18}}\label{samir2}\hfill{~}
  \begin{enumerate}
\item Let $t=\frac{O(\log n)}{\log\log n}$ and $B$ be a $t\times t$
  integer matrix with entries bounded by $p^m$ and $q$ be an
  $O(\log\log n)$ bit prime number. Then both $\det(B) (\Mod q)$ and
  $B^{-1}$ over $\F_q$ can be computed in $\AC^0$.
\item Furthermore, by Chinese remaindering, $\det(B)$ and hence
  $B^{-1}$, if it exists, can both be computed in $\AC^0$ by applying
  the first part for several distinct primes $q_i$ and different
  submatrices.
\end{enumerate}  
\end{lemma}

\paragraph{\textbf{Processing Bulk Insertions and Deletions}}

We recall that $S\subseteq G$ is the current generating set. We will
additionally maintain a subset $T\subseteq S$ of size at most $\log n$
that also generates the same group: $\angle{S}=\angle{T}$. 
 
Suppose $\hat{T}$ is the set of insertions to $S$, where
$|\hat{T}|=t=O(\log n/\log\log n)$. Thus, the system of linear
equations $Ax=\bar{b}~(\Mod p^m)$ is now modified to
\[
  [A | \hat{A}]x=\bar{b}~(\Mod p^m),
\]  
where $\bar{b}$ is the integer vector corresponding to a $g\in G$
whose membership we want to test in $\angle{T\cup \hat{T}}$.

Similarly, suppose $\hat{T}$ is the set of deletions to $S$ which are contained
in $T$. Let $T'=T\setminus \hat{T}$ and $A'$ be the corresponding column
vectors. The system of linear equations is then modified to
\[
  A'x = \bar{b}~(\Mod p^m).
\]

We note that in both modified linear equations above, the original
coefficient matrix has been modified in at most $t$ columns. Thus,
Lemmas~\ref{SMW} and \ref{mdet-lem} are applicable. With these we can
update the data computed by Lemma~\ref{samir1} in $\AC^0$. It can be
recomputed in $\AC^0$ by Lemma~\ref{samir2} as at most $O(\log
n/\log\log n)$ columns are modified in any submatrix, thinking of the
new columns as modifications of zero columns. Furthermore the
recomputations involves computing the determinant and inverse of
matrices of dimension at most $t=O(\log n/\log\log n)$, where those
matrices have integer entries given as input in unary (because each of
them is at most $p^m$ in magnitude). A crucial difficulty in the
application of Lemmas~\ref{SMW} and \ref{mdet-lem} is that if a
submatrix $M$ of $A$, whose inverse/determinant we need to update, may
itself not be invertible. We can deal with this by maintaining the
data computed by Lemma~\ref{samir1} for the \emph{invertible} matrices
$\xi I-M$, for all submatrices $M$ of $A$, where $\xi$ is an
indeterminate.  Lemma~\ref{lem:detPolylogInFOLL} and parts 1 and 2 of
Lemma~\ref{samir1} can be applied \textit{mutatis mutandis} to
matrices $\xi I-M$ (for the submatrices $M$ of $\tilde{A}$). The
determinant of $\xi I-M$ will be a degree $r$ polynomial in $\xi$ and
$(\xi I-M)^{-1}$ will have entries that are rational functions
$f(\xi)/g(\xi)$ where $f$ and $g$ are of degree at most $r$, where
$r=O(\log n)$. Consider Equations~\ref{mil-eqn} and \ref{mdl-eqn}
applied to $\xi I - M$ instead of $M$. Notice that $\det(M+UV^T)$ is
the constant term of $\det(M + UV^T -\xi I)$ which we can compute in
$\AC^0$, essentially by Lemma~\ref{samir2}. Similarly, by
Lemma~\ref{samir2} the inverse $(M+UCV^T)^{-1}$, if it exists, can be
computed in $\AC^0$ from Equation~\ref{mil-eqn} applied to $\xi I - M$.

\paragraph{\textbf{Continual rebuilding of data structure}}~We note
that the data required by the algorithm to answer membership queries
in $\AC^0$ are:
\begin{itemize}
  \item A generating set $T$ of size at most $\log n$ such that
    $\angle{T}=\angle{S}$.
  \item An $\ell\times |T|$ dimensional matrix $A$ corresponding to
    the generating set along with the inverses and determinants of all
    its submatrices as computed by Lemma~\ref{samir1}.
\end{itemize}

With the above the algorithm can check feasibility of
$Ax=\bar{b}~(\Mod p^m)$ in $\AC^0$ and hence answer membership
queries. However, there are $O(\log n/\log\log n)$
insertions/deletions in every step to the generating set
$T$. Cumulatively, over some $k$ steps the modified generating set,
say $T\cup \hat{T}$ can be of size $\log n + k\cdot O(\log n/\log\log
n)$ which is $O(\log n)$ for $k=O(\log\log n)$. Thus, in order to keep
the generating set small, we will need to use Lemma~\ref{preproc} to
compute a subset $\tilde{T}$ of $T\cup \hat{T}$ of size at most $\log
n$ which generates the same group as $T\cup \hat{T}$. 
%The remaining
%data structure, as in Lemma~\ref{samir1}, is already there as
%$\tilde{T}$ is a subset of $T\cup \hat{T}$.

To describe the rebuilding process more precisely, at each time
instant $i$, if the current generating set is $T\cup \hat{T}$ we run
the $\FOLL$ algorithm of Lemma~\ref{preproc} to compute the subset
$\tilde{T}$. In time instant $i+\log\log n$ we will have $\tilde{T}$
along with the modifications $T'$ accumulated in $\log\log n$ time
interval from $i$ to $i+\log\log n$. During this time period, we can answer the membership query at every instant under insertion/deletion of $O(\log n /\log\log n)$ by updating all the submatrices' inverses and determinants in $\AC^0$ by Lemmas~\ref{SMW} and~\ref{mdet-lem} respectively. Thus, at time instant $i+\log\log n$ we would have all the submatrix information pertaining to the 
set $T\cup\hat{T}\cup T'$, and hence, also for its subset $\tilde{T}\cup T'$. More precisely, if the bulk insertions
and deletions are bounded by $c(\log n/\log\log n)$ then
$|\tilde{T}\cup T'|\le (c+1)\log n$ at any time instant. Furthermore,
notice that at any point in time there are $\log\log n$ parallel
threads of computation carrying out this rebuilding of data, one each
corresponding to a time instant in the $\log\log n$ window.

\begin{comment}
To describe the rebuilding process more precisely, at each time
instant $i$, if the current generating set is $T\cup \hat{T}$ we run
the $\FOLL$ algorithm of Lemma~\ref{preproc} to compute the subset
$\tilde{T}$. In time instant $i+\log\log n$ we will have $\tilde{T}$
along with the modifications $T'$ accumulated in $\log\log n$ time
interval from $i$ to $i+\log\log n$. All other required data to answer
a membership query at time instant $i+\log\log n$ for the generating
set $\tilde{T}\cup T'$ would also be ready by Lemma~\ref{samir2}
(using Lemmas \ref{SMW} and \ref{mdet-lem}), where $\tilde{T}\cup T'$
is also of size $O(\log n)$. More precisely, if we the bulk insertions
and deletions are bounded by $c(\log n/\log\log n)$ then
$|\tilde{T}\cup T'|\le (c+1)\log n$ at any time instant. Furthermore,
notice that at any point in time there are $\log\log n$ parallel
threads of computation carrying out this rebuilding of data, one each
corresponding to a time instant in the $\log\log n$ window.
\end{comment}
To summarize we have shown the following theorem.

\begin{theorem}\label{thm2}
  There is a deterministic dynamic $\AC^0$ algorithm for the abelian
  Cayley Group Membership problem, $\CGM$, that supports $O(\log
  n/\log\log n)$ insertions and deletions to the generating set.
\end{theorem}

\section{Dynamic Abelian group isomorphism}\label{sec:AbIso}

Let $G_1$ and $G_2$ be abelian groups, each given a multiplication
table as input, say $T_1$ and $T_2$, respectively.  Let $S_1\subseteq
G_1$ and $S_2\subseteq G_2$ be subsets. In the static setting, there
is a simple polynomial time algorithm for checking if $\angle{S_1}$
and $\angle{S_2}$ are isomorphic: it suffices to list out the two
subgroups $\angle{S_1}$ and $\angle{S_2}$, check they have the same
order $n$, and check for each factor $k$ of $n$ that the number of
elements of order $k$ in the two subgroups $\angle{S_1}$ and
$\angle{S_2}$ is the same.\footnote{Two finite abelian groups are
  isomorphic iff for each positive integer $k$ the number of elements
  of order $k$ in the two groups coincide \cite{Hall}.}
  
We give a $\DynAC^0$ algorithm the \emph{dynamic version} of abelian
group isomorphism that supports insertions and deletions to both $S_1$
and $S_2$.

Let $n_1=\prod_{x\in S_1}o(x)$ and $n_2=\prod_{y\in S_2}o(y)$, where
the orders $o(x), o(y), x\in S_1, y\in S_2$ can be pre-computed for
the elements of the two groups $G_1$ and $G_2$. Let $n_1=\prod_{i=1}^r
p_i^{a_i}$ and $n_2=\prod_{i=1}^r p_i^{b_i}$ be their prime
factorizations. We can assume both $n_1$ and $n_2$ have the same prime
factors. Otherwise, $\angle{S_1}$ and $\angle{S_2}$ are not
isomorphic. Let $n_{1i}=n_1/p_i^{a_i}$ and $n_{2i}=n_2/p_i^{b_i}$ and
$S_{1i}=\{x^{n_{1i}}\mid x\in S_1\}$ and $S_{1i}=\{y^{n_{1i}}\mid y\in
S_2\}$ for $1\le i\le r$. Since a finite abelian group is a direct product
of its (unique) Sylow subgroups  we have

\begin{proposition}
 The groups $\angle{S_1}$ and $\angle{S_2}$ are isomorphic iff their
 $p_i$-Sylow subgroups $\angle{S_{1i}}$ and $\angle{S_{2i}}$ are
 isomorphic.
\end{proposition}

Thus, as argued in Sections~\ref{sec4} and \ref{sec4a}, it suffices to
solve the problem for abelian $p$-groups. Henceforth, we assume both
$\angle{S_1}$ and $\angle{S_2}$ are $p$-groups. The following lemma
from Mckenzie and Cook's work \cite{McC}, paraphrased in our context,
is useful for our algorithm.

\begin{lemma}{\rm\cite[Proposition 6.4]{McC}}\label{mckenzie}
  Let $\angle{S_1}\le G_1$ and $\angle{S_2}\le G_2$ be abelian
  $p$-groups, and $k$ be largest positive integer such that $p^k\le
  \max\{|G_1|,|G_2|\}$. For $1\le j\le k$ let $S_{1j}=\{x^{p^j}\mid
    x\in S_1\}$ and $S_{2j}=\{y^{p^j}\mid y\in S_2\}$. Then
    $\angle{S_1}$ and $\angle{S_2}$ are isomorphic if and only if
    $|\angle{S_{1j}}|=|\angle{S_{2j}}|$ for $1\le j\le k$.
\end{lemma}

Effectively, the above lemma is a reduction from abelian group
isomorphism to abelian $\CGM$. Thus, as observed in \cite{CT12}, in
the static setting we can note that the above lemma immediately shows
that abelian group isomorphism problem we consider can be solved by
$\AC(\log\log n)$ circuits with majority gates. This is by applying
the Barrington et al algorithm \cite{BKLM} to enumerate the subgroups
$\angle{S_{1j}}$ and $\angle{S_{2j}}$ in $\AC(\log\log n)$ for each
$1\le j\le k$ and then comparing their orders (for which majority
gates are required).

%\begin{proposition}{\rm\cite{CT12}}
%  Abelian group isomorphism $\FOLL$.
%\end{proposition}

Our strategy for the dynamic version is also based on
Lemma~\ref{mckenzie} because we can apply the results for abelian
$\CGM$ shown in Sections \ref{sec4} and \ref{sec4a}.

In the dynamic setting, where we have insertions and deletions to the
generating sets $S_1,S_2$, we will use the same data structures
developed in Section~\ref{sec4} (for supporting $\polylog(n)$
insertions and deletions) and Section~\ref{sec4a} (for supporting
$\log n/\log\log n$ insertions and deletions) for the abelian $\CGM$
problem but now in parallel for all the generating sets
$\angle{S_{1j}}$ and $\angle{S_{2j}}$ for $1\le j\le k$.

In order to compute $|\angle{S_{1j}}|$ and $|\angle{S_{2j}}|$ from
membership queries the following lemma, from \cite{McC}, is useful.

 \begin{lemma}{\rm\cite[Proposition 6.6]{McC}}
	Let $H=\angle{g_1,\ldots,g_r}$ be a finite abelian
        $p$-group. Then, $|H| = t_1 t_2\ldots t_r$ where $t_j$ is the
        least positive integer such that $g_j^{t_j} \in
        \angle{g_{j+1},\ldots,g_r}$ for $1\leq j \leq r$.
\end{lemma}     

 In the above lemma, as $H$ is a $p$-group notice that each $t_j$ is a
 power of $p$. We will be applying this lemma to groups
 $\angle{S_{1j}}$ and $\angle{S_{2j}}$. As $|\angle{S_{1j}}|\le n_1\le
 |G_1|$ and $|\angle{S_{2j}}|\le n_2\le |G_2|$, and both $n_1$ and
 $n_2$ are logarithmic size in binary, for these groups $H$ 
 at most logarithmically many of the integers $t_i$ are more than
 $1$. Letting $t_j=p^{r_j}$, computing the product $\prod_j
 t_j=p^{\sum_j r_j}$ amounts to adding at most logarithmically many
 $r_j$, each logarithmically bounded. As already observed, such tiny
 additions can be computed in $\AC^0$.

To summarize, we have shown the following.

\begin{theorem}
\begin{enumerate}
\item There is a randomized $\DynAC^0$ algorithm for abelian group
  isomorphism that supports $\polylog(n)$ insertions and deletions at
  each step to the generating sets of the two groups.
\item There is a deterministic $\DynAC^0$ algorithm for abelian group
  isomorphism that supports $\log n/\log\log n$ insertions and
  deletions at each step to the generating sets of the two groups.
\end{enumerate}
\end{theorem}

\section{Making the multiplication table dynamic}\label{sec6}

We have assumed so far that the overall group $G$ (or monoid) is
unchanged and only the generating set for the $\CGM$ problem is
dynamic. Suppose now that the entries of the multiplication table of
$G$ can be modified dynamically. When the table's entries change, it
may no longer represent a group (or a monoid). The binary operation
$*: G\times G\to G$ is just a \emph{magma}, in general. However, we
can show that the dynamic algorithms for abelian $\CGM$ still hold,
with the proviso that the membership query answers are correct only
when the magma is actually an abelian group.

The main property we use here is that at most one group has its
multiplication table within linear (i.e.\ $O(n)$) edit distance from
the multiplication table of an $n$-element magma $G$.\footnote{By the
  edit distance between multiplication tables $op_1:G\times G\to G$
  and $op_2:G\times G\to G$ we mean the number of pairs $(a,b)\in
  G\times G$ such that $op_1(a,b)\ne op_2(a,b)$.} Moreover, from the
magma multiplication table we can \emph{decode} this unique group in
$\AC^0$. We note that Erg\"{u}n et al \cite{EKKRV} have shown stronger
results for this problem in the context of spot checkers; they give
randomized self-correction algorithms for a variety of problems.
However, for a self-contained presentation, we include a simple proof
of a weaker statement that suffices for our purpose.

\begin{lemma}\label{decode}
Let $M_G$ denote the multiplication table of the group $G$. Suppose
$M$ is a multiplication table obtained from $M_G$ by changing at most
$\delta n$ many entries of $M_G$, for $\delta < 1/13$. Then there is
an $\AC^0$ circuit that takes $M$ as input and outputs $M_G$.
\end{lemma}

\begin{proof}
  For each $x_i\in G$ the row of $x_i$ in $M$ has at most $\delta n$
  errors in it. Thus, the row for the identity element, say $x_1=e$ is
  uniquely determined, because $x_1\in G$ is the unique element with
  $x_1x_j=x_j$ for majority of $j\in [n]$.

  For each $z\in G$ there is a unique inverse $z^{-1}\in G$ and
  $zz^{-1}=e=z^{-1}z$. That means in $M_G$ there are exactly $n$
  occurrences of $e$ in the multiplication table. Therefore, in $M$
  there are at most $(1+\delta)n$ occurrences of $e$ and at least
  $(1-\delta)n$ occurrences of $e$.

  Let $S=\{(z,w)\mid z,w\in G, z*w=e\}$, where $*$ is the product
  operation in the table $M$.  Then
  \[
    (1-\delta)n\le |S| \le (1+\delta)n.
  \]

  Thus, for at least $(1-2\delta)n$ pairs $(z,w)\in S$ we have
  $z*w=zw=e$ in $G$.

  Now, in order to recover the correct value of the product $x_ix_j$,
  we look up the products $(x_i*z)*(w*x_j)$ in the table $M$. Then
  we have
\begin{itemize}
\item $|\{z\in G\mid x_i*z\ne x_iz\}| \le \delta n$.
\item $|\{w\in G\mid w*x_j\ne wx_j\}| \le \delta n$.
\item $|\{(z,w)\in S\mid z*w\ne zw\}| \le \delta n$.
\item $|\{(z,w)\in S\mid (x_iz)*(wx_j)\ne (x_iz)(wx_j)\}| \le \delta n$.  
\end{itemize}

Thus, for at least $(1-6\delta)n$ pairs $(z,w)\in S$ we have
$(x_i*z)*(w*x_j)=(x_iz)(wx_j)=x_i(zw)x_j=x_ix_j$.

If we choose $\delta<1/13$ then the number of such pairs $(z,w)\in S$
is more than $7n/13$. By approximate majority which can be computed
in $\AC^0$ \cite{Ajtai,Ajtai2}, we can find this correct value of $x_ix_j$. We can thus
recover the entire table $M_G$ in $\AC^0$.  
\end{proof}

The above lemma suggests the following simple dynamic algorithm for
the abelian $\CGM$ problem that supports $\polylog(n)$ changes to the
group multiplication table, as well as bulk insertions/deletions to
the generating set (as discussed in Sections \ref{sec4} and
\ref{sec4a}).

\begin{enumerate}
\item Let $M$ be the current multiplication table. We apply the
  $\AC^0$ algorithm of Lemma~\ref{decode} to decode $M$. Let $G$ be
  the resulting table.

%%  The algorithm also
%%  keeps the multiplication table of the previous abelian group $G$ for
%%  which we keep the data structures as described in Section~\ref{sec4}
%%  (or the one in Section~\ref{sec4a}).

%%\item If the decoded table gives the same abelian group $G$ then we
%%  process the additions/deletions to the generating set $S$ and
%%  membership query as per the algorithm in Section~\ref{sec4} or
%%  \ref{sec4a}.

\item If the decoded table does not give an abelian group, the query
  answers can be arbitrary (but consistent which can be ensured by
  remembering the answers to queries already made).

\item Suppose the decoded table gives an abelian group $G'$. If $G'\ne
  G$ then for the next $\log n$ steps we rebuild the static data
  structure for $G'$ in $\AC(\log n)$. We can answer any membership
  queries, occurring in this window of $\log n$ time steps,
  arbitrarily. After $\log n$ steps we can replace $G$ with $G'$ and
  its data structure and continue.
\end{enumerate}  

\begin{comment}
\begin{enumerate}
\item Let $M$ be the current multiplication table. The algorithm also
  keeps the multiplication table of the previous abelian group $G$ for
  which we keep the data structures as described in Section~\ref{sec4}
  (or the one in Section~\ref{sec4a}).
\item We apply the $\AC^0$ algorithm of Lemma~\ref{decode} to decode
  $M$. Let $G$ be the resulting table.
\item If the decoded table gives the same abelian group $G$ then we
  process the additions/deletions to the generating set $S$ and
  membership query as per the algorithm in Section~\ref{sec4} or
  \ref{sec4a}.
\item If the decoded table does not give an abelian group, the query
  answers can be w.r.t. $G$ (which are arbitrary in a sense).
\item if the decoded table gives a new abelian group $G'$ then for the
  next $\log n$ steps we rebuild the static data structure for $G'$ in
  $\AC(\log n)$.  We can answer any membership queries, occurring in
  this window of $\log n$ time steps, arbitrarily. After $\log n$
  steps we can replace $G$ with $G'$ and its data structure and
  continue.
\end{enumerate}
\end{comment}

In summary, we have the following.

\begin{theorem}
  There is a randomized $\DynAC^0$ algorithm that supports $O(n/\log
  n)$ changes to the multiplication table and $\polylog(n)$
  insertions/deletions to the generating set, with the proviso that
  when the multiplication table decodes to an abelian group the
  membership queries are answered with respect to it, and when it does
  not decode to an abelian group then the query answers could be
  incorrect.  There is also a deterministic $\DynAC^0$ algorithm that
  supports $O(n/\log n)$ changes to the multiplication table and $\log
  n/\log\log n$ insertions/deletions to the generating set, with the
  same proviso as described above.
\end{theorem}

%\bibliographystyle{splncs04}% the mandatory bibstyle
%\bibliography{main1}

%\appendix
%\input{abIsoApp}
%\input{igs-foll}
\end{document}